\documentclass[11pt]{amsart}

\usepackage[foot]{amsaddr}
\usepackage{amsmath}
\usepackage{amssymb}
\usepackage{amsthm}
\usepackage{fullpage}
\usepackage{hyperref}
\usepackage[english]{isodate}
\usepackage{tikz}
\usetikzlibrary{arrows}

\tikzset{>=stealth}

\theoremstyle{plain}

\newtheorem{lemma}{Lemma}
\newtheorem{theorem}{Theorem}

\theoremstyle{definition}

\renewcommand{\geq}{\geqslant}
\renewcommand{\leq}{\leqslant}

\newcommand{\ints}{\mathbb{Z}}

\renewcommand{\leq}{\leqslant}
\renewcommand{\geq}{\geqslant}

\DeclareMathOperator{\B}{B}
\DeclareMathOperator{\K}{K}

\newcommand{\Nout}{N^{\text{out}}}
\newcommand{\Nin}{N^{\text{in}}}

\title{On the power domination number of de Bruijn and Kautz digraphs}

\author{Cyriac Grigorious \and Thomas Kalinowski \and Joe Ryan \and Sudeep Stephen}
\address{The University of Newcastle, Australia}

\date{\today}

\begin{document}

\begin{abstract}
  Let $G=(V,A)$ be a directed graph without parallel arcs, and let $S\subseteq V$ be a set of
  vertices. Let the sequence $S=S_0\subseteq S_1\subseteq S_2\subseteq\cdots$ be defined as follows:
  $S_1$ is obtained from $S_0$ by adding all out-neighbors of vertices in $S_0$. For $k\geqslant 2$,
  $S_k$ is obtained from $S_{k-1}$ by adding all vertices $w$ such that for some vertex
  $v\in S_{k-1}$, $w$ is the unique out-neighbor of $v$ in $V\setminus S_{k-1}$. We set
  $M(S)=S_0\cup S_1\cup\cdots$, and call $S$ a \emph{power dominating set} for $G$ if $M(S)=V(G)$.
  The minimum cardinality of such a set is called the \emph{power domination number} of $G$. In this
  paper, we determine the power domination numbers of de Bruijn and Kautz digraphs.
\end{abstract}

\maketitle

\section{Introduction}
Let $G = (V, A)$ be a directed graph. For a vertex $i\in V$ let $\Nin(i)$ and $\Nout(i)$ denote its in-
and out-neighborhood, respectively, i.e.,
\[\Nin(i)=\{j\in V\ :\ (j,i)\in A\},\qquad \Nout(i)=\{j\in V\ :\ (i,j)\in A\}.\]
For a node set $S$, we use the corresponding notation
\[\Nin(S)=\bigcup_{i\in S}\Nin(i),\qquad \Nout(S)=\bigcup_{i\in S}\Nout(i).\] Let $G $ be a directed graph
and $S$ a subset of its vertices. Then we denote the set monitored by $S$ with $M(S)$ and define it
as $M(S)=S_0\cup S_1\cup\cdots$ where the sequence $S_0,S_1,\ldots$ of vertex sets is defined by
$S_0=S$, $S_1=\Nout(S)$, and
\[S_k=S_{k-1}\cup\left\{w\,:\,\{w\}=\Nout(v)\cap\left(V\setminus S_{k-1}\right)\text{ for some } v \in
    S_{k-1}\right\}.\]
A set $S$ is called a \emph{power dominating set} of $G$ if $M(S)=V(G)$ and
the minimum cardinality of such a set is called the \emph{power domination number} denoted as
$\gamma_p(G)$.

The undirected version of the power domination problem was introduced in~\cite{Haynes2002a}. The
problem was inspired by a problem in electric power systems concerning the placements of
phasor measurement units. The directed version of the power domination problem was introduced as a
natural extension in~\cite{AazamiStilp2009} where a linear time algorithm was presented for digraphs whose
underlying undirected graph has bounded treewidth. Good literature reviews
on the power domination problem can be found in~\cite{Chang2012, Dorbec2008, Stephen2015}

A closely related concept is zero forcing which was introduced for undirected graphs by the
\emph{AIM Minimum Rank -- Special Graphs Work Group} in~\cite{AIM-2008-Zeroforcingsets} as a tool to
bound the minimum rank of matrices associated with the graph $G$. This notion was extended to
digraphs in~\cite{Barioli2009} with the same motivation. For a red/blue coloring of the vertex set
of a digraph $G$ consider the following color-change rule: a red vertex $w$ is converted to blue if
it is the only red out-neighbor of some vertex $u$. We say $u$ forces $w$ and denote this by
$u \rightarrow w$. A vertex set $S\subseteq V$ is called \emph{zero-forcing} if, starting with the
vertices in $S$ blue and the vertices in the complement $V\setminus S$ red, all the vertices can be
converted to blue by repeatedly applying the color-change rule. The minimum cardinality of a
zero-forcing set for the digraph $G$ is called the \emph{zero-forcing number} of $G$, denoted by
$Z(G)$.  Since its introduction the zero-forcing number has been studied for its own sake as an
interesting graph
invariant~\cite{BarioliBarrettFallatEtAl2010,BarioliBarrettFallatHallHogbenShaderDriesscheHolst-2012-ParametersRelatedto,BermanFriedlandHogbenEtAl2008,EdholmHogbenHuynhEtAl2012,Lu.etal_2015_Proofconjecturezero}. In~\cite{HogbenHuynhKingsleyMeyerWalkerYoung-2012-Propagationtimezero},
the \emph{propagation time} of a graph is introduced as the number of steps it takes for a zero
forcing set to turn the entire graph blue. Physicists have independently studied the zero forcing
parameter, referring to it as the \emph{graph infection number}, in conjunction with the control of
quantum systems~\cite{Severini2008}.

Recently, Dong \emph{et al.} (2015)~\cite{Dong2015a} investigated the domination number of
generalized de Bruijn and Kautz digraphs. Kuo \emph{et al.}(2015)~\cite{Kuo2015} gave an upper bound
for power domination in undirected de Bruijn and Kautz graphs. In this paper we study the directed
versions, i.e., the zero forcing number and power domination number of de Bruijn and Kautz
digraphs. Due to their attractive connectivity features these digraphs have been widely studied as a
topology for interconnection networks~\cite{Huang2006}, and some generalizations of these digraphs
were proposed~\cite{Imase1983}.

Section~\ref{sec:notation} contains some notation and precise statements of our main result. In
Section~\ref{sec:debruijn} we determine the power domination number and zero forcing number for de Bruijn
digraphs. In Section~\ref{sec:kautz} we determine the power domination number and zero forcing number for
Kautz digraphs.

\section{Notations and main result}\label{sec:notation}

We give an interpretation of the power domination problem and zero forcing problem as a set cover
problem. We call a vertex set $W$ \emph{strongly critical} if there is no vertex in $G$ which has
exactly one out neighbor in $W$. We call a vertex set $W$ \emph{weakly critical} if there is no
vertex outside $W$ which has exactly one out-neighbor in $W$. If $W$ is strongly (weakly) critical,
but no proper subset of $W$ is strongly (weakly) critical, then we call $W$ \emph{minimal strongly
  (weakly) critical}.

Note that a vertex set $S$ is a zero forcing set if and only if $S \cap W \neq \emptyset$ for every
strongly critical set $W\subseteq V$. Similarly, $S$ is a power dominating set if and only if $\Nout(S) \cap W \neq \emptyset$ for every
weakly critical set $W\subseteq V$, and therefore
\begin{align*}
Z(G) &= \min\left\{\lvert S\rvert\ : \ S \cap W \neq \emptyset \text{ for every strongly critical set}\ W \subseteq V \right\},\\
\gamma_{p}(G) &= \min\left\{\lvert S\rvert\ :\ (S\cup \Nout_G(S)) \cap W  \neq \emptyset \textrm{ for every weakly critical set } W\subseteq V \right\}.  
\end{align*}
For an integer $d\geq 2$, let $\ints_d=\{0,1,\ldots,d-1\}$ denote the cyclic group of order $d$. The
de Bruijn digraph, denoted $\B(d,n)$, with parameters $d\geq 2$ and $n\geq 2$ is defined to be the
graph $G=(V,A)$ with vertex set $V$ and arcs set $A$ where
\begin{align*}
V &= \ints_d^n=\left\{\left(a_1,\ldots,a_n\right)\ :\ a_i\in\ints_d\text{ for }i=1,\ldots,n\right\},\\
A &= \left\{\left((a_1,a_2,\ldots,a_n),(a_2,\ldots,a_n,b)\right)\ :\ (a_1,a_2,\ldots,a_n)\in V,\ b\in\ints_d\right\}.
\end{align*}
The Kautz digraph, denoted $\K(d,n)$, with parameters $d\geq 2$ and $n\geq 2$ is defined to be the
graph $G=(V,A)$ with vertex set $V$ and arcs set $A$ where
\begin{align*}
V &=\left\{\left(a_1,\ldots,a_n\right)\ :\ a_i\in\ints_{d+1},\ a_i\neq a_{i+1}\right\}\\
A &= \left\{\left((a_1,a_2,\ldots,a_n),(a_2,\ldots,a_n,b)\right)\ :\ (a_1,a_2,\ldots,a_n)\in V,\ b\in\ints_{d+1}\setminus\{a_n\}\right\}.
\end{align*}

Our main results are the following theorems.

\begin{theorem}\label{main:thm:debruijn}
  Let $G$ be a de Bruijn digraph with parameters $d,n\geq 2$. Then the zero forcing number and power domination number of $G$
  are $(d-1) d^{n-1}$ and $(d-1)d^{n-2}$, respectively.
\end{theorem}

\begin{theorem}\label{main:thm:kautz}
  Let $G$ be a Kautz digraph with parameters $d\geq 2$ and $n\geq 3$. Then, the zero forcing number
  and power domination number of $G$ are $(d-1)(d+1)d^{n-2}$ and $(d-1)(d+1)d^{n-3}$,
  respectively.
\end{theorem}

\section{The power domination number of de Bruijn digraphs}\label{sec:debruijn}
In this section we prove Theorem~\ref{main:thm:debruijn}. Let us define the sets
\[X(a_1,\ldots,a_{n-1}) = \{(a_1,\ldots,a_{n-1},\alpha)\ :\ \alpha\in\ints_d\}\] which partition the
vertex set $V$ into $d^{n-1}$ sets of size $d$. Furthermore, $\Nout(v)= X(a_1,\ldots,a_{n-1})$ for
every vertex $v$ of the form $(\alpha,a_1,a_2,\ldots,a_{n-1})$.
\begin{lemma}\label{lb:zfdebruijn}
  Let $G$ be a de Bruijn digraph with parameters $d$ and $n$. Then $Z(G) \geq (d-1) d^{n-1}$.
\end{lemma}
\begin{proof}
  Every $2$-element subset of each of the sets $X(a_1,\ldots,a_{n-1})$ is strongly critical, and
  therefore, any zero forcing set $S$ needs to intersect $X(a_1,\ldots,a_{n-1})$ in at least $d-1$
  elements, and the result follows.
\end{proof}

\begin{lemma}\label{ub:zfdebruijn}
  Let $G$ be a de Bruijn digraph with parameters $d$ and $n$. Then $Z(G) \leq (d-1) d^{n-1}$.
\end{lemma}
\begin{proof}
  Consider the vertex set $S=\{(a_1,\ldots,a_{n-1},a_n)\in V\,:\,a_1\neq a_n\}$. To show that $S$ is
  a zero forcing set, it is sufficient to verify that each vertex $v=(a_1,\ldots,a_{n-1},a_n)$ is
  either in $S$ or is the unique out-neighbor in $V\setminus S$ for some vertex $w$. If
  $a_1 \neq a_n$, then $v \in S$. If $a_1 = a_n$, then for any vertex of the
  form $w=(\beta,a_1,\ldots,a_{n-1})$, $v$ is the only neighbor of $w$ in $V\setminus S$.
\end{proof}
Lemmas~\ref{lb:zfdebruijn} and~\ref{ub:zfdebruijn} imply the first statement of
Theorem~\ref{main:thm:debruijn}. In order to prove the second part of this theorem we recall that
$S\subseteq V$ is a power dominating set if and only if $S\cup\Nout(S)$ intersects every weakly
critical set. In particular, it is necessary that $\lvert(S\cup\Nout(S))\cap
X(a_1,\ldots,a_{n-1})\rvert\geqslant d-1$ for every $(a_1,\ldots,a_{d-1})\in\ints_d^{n-1}$.

\begin{lemma}\label{lb:pddebruijn}
  Let $G$ be a de Bruijn graph with parameters $d$ and $n$. Then every power dominating set has size at least $(d-1)d^{n-2}$.
\end{lemma}
\begin{proof}
  Let $S$ be a power-dominating set, suppose $\lvert S\rvert<(d-1)d^{n-2}$ and set
  $Z=S\cup \Nout(S)$. We have
  \[(Z\setminus S)\cap X(a_1,\ldots,a_{n-1})\neq\emptyset\implies X(a_1,\ldots,a_{n-1})\subseteq
    Z.\]
  For $k=0,1,\ldots,d$, we set $\alpha_k=\#\{(a_1,\ldots,a_{n-1})\ :\ \lvert S\cap X(a_1,\ldots,a_{n-1})\rvert=k\}$, and get
  \[\lvert S\rvert=\alpha_1+2\alpha_2+\cdots+(d-1)\alpha_{d-1}+d\alpha_d.\]
  Now let $I_0=\{(a_1,\ldots,a_{n-1})\ :\ X(a_1,\ldots,a_{n-1})\subseteq Z\}$. Then
  \[\lvert I_0\rvert\leqslant \lvert
    S\rvert+\alpha_d=\alpha_1+2\alpha_2+\cdots+(d-1)\alpha_{d-1}+(d+1)\alpha_d.\]
  For $(a_1,\ldots,a_{n-1})\notin I_0$ we must have $\lvert Z\cap X(a_1,\ldots,a_{n-1})\rvert=d-1$,
  and this implies that $\lvert S\cap X(a_1,\ldots,a_{n-1})\rvert=d-1$. We conclude
  $\lvert I_0\rvert+\alpha_{d-1}\geqslant d^{n-1}$. Therefore
  \[\alpha_1+2\alpha_2+\cdots+(d-2)\alpha_{d-2}+d\alpha_{d-1}+(d+1)\alpha_d\geqslant d^{n-1},\]
  and together with $\lvert S\rvert<(d-1)d^{n-2}$ this yields
  \[\alpha_{d-1}+\alpha_d>d^{n-1}-(d-1)d^{n-2}=d^{n-2}.\]
  But then $\lvert S\rvert\geqslant (d-1)(\alpha_{d-1}+\alpha_d)>(d-1)d^{n-2}$, which is the required contradiction.
\end{proof}

We define a set $S\subseteq V$ by
\begin{equation}\label{eq:pd_set_debruijn}
S=
\begin{cases}
\{(0,1),(0,2),\ldots,(0,d-1)\} & \text{if }n=2,\\
\{(a_1,a_2,a_3)\in V\ :\ a_2=a_1,\,a_3\neq a_1\} & \text{if }n=3,\\
\{\left(a_1,\ldots,a_n\right)\in V\ :\ a_{n-1}=a_1+a_{n-2},\ a_n\neq a_1+a_2+a_{n-2}\} & \text{if }n\geqslant 4.    
\end{cases}
\end{equation}
Note that $\lvert S\rvert=(d-1)d^{n-2}$. The construction of the set $S$ defined in~(\ref{eq:pd_set_debruijn}) can be visualized by arranging
the vertices of $G$ in a $d^2\times d^{n-2}$-array where the rows are indexed by pairs
$(a_{n-1},a_n)$ and the columns are indexed by $(n-2)$-tuples $(a_1,\ldots,a_{n-2})$. Then column
$(a_1,\ldots,a_{n-2})$ is the the union of the $d$ sets $X(a_1,\ldots,a_{n-2},a_{n-1})$ over
$a_{n-1}\in\ints_d$, and the set $S$ contains $d-1$ elements from each column. More precisely, the
intersection of $S$ with column $(a_1,\ldots,a_{n-2})$ is
\[X(a_1,\ldots,a_{n-2},a_1+a_{n-2})\setminus\{(a_1,\ldots,a_{n-2},a_1+a_{n-2},a_1+a_2+a_{n-2})\}.\]
In Figure~\ref{fig:debruijn} this is illustrated for two columns with $d=5$ and $n=7$.
\begin{figure}[htb]
	\centering
	\begin{tikzpicture}[yscale=.3]
	\node at (9,27) {$(3,1,0,2,4)$};
	\node at (3,27) {$(1,3,4,4,2)$};
	\node at (1,2.5) {$a_6=4$};
	\node at (1,7.5) {$a_6=3$};
	\node at (1,12.5) {$a_6=2$};
	\node at (1,17.5) {$a_6=1$};
	\node at (1,22.5) {$a_6=0$};
	\foreach \y in {1,...,25} \node[draw=black,circle,fill=black,inner sep=.8] at (3,\y) {};
	\foreach \y in {1,...,25} \node[draw=black,circle,fill=black,inner sep=.8] at (9,\y) {};
	\draw (0,5.5) -- (12,5.5);
	\draw (0,10.5) -- (12,10.5);
	\draw (0,15.5) -- (12,15.5);
	\draw (0,20.5) -- (12,20.5);
	\draw (2.9,20.7) rectangle (3.1,25.3);
	\draw (2.9,15.7) rectangle (3.1,20.3);
	\draw (2.9,10.7) rectangle (3.1,15.3);
	\draw (2.9,0.7) rectangle (3.1,5.3);
	\draw (8.9,20.7) rectangle (9.1,25.3);
	\draw (8.9,15.7) rectangle (9.1,20.3);
	\draw (8.9,5.7) rectangle (9.1,10.3);
	\draw (8.9,0.7) rectangle (9.1,5.3);
	\node[anchor=west] at (3.2,24.2) {{\small $X(1,3,4,4,2,0)$}};
	\node[anchor=west] at (3.2,22.2) {{\small $=\Nout(3,1,3,4,2,2,0)$}};
	\node[anchor=west] at (3.2,19.2) {{\small $X(1,3,4,4,2,1)$}};
	\node[anchor=west] at (3.2,17.2) {{\small $=\Nout(3,1,3,4,2,2,1)$}};
	\node[anchor=west] at (3.2,14.2) {{\small $X(1,3,4,4,2,2)$}};
	\node[anchor=west] at (3.2,12.2) {{\small $=\Nout(3,1,3,4,2,2,2)$}};
	\node[anchor=west] at (3.2,4.2) {{\small $X(1,3,4,4,2,4)$}};
	\node[anchor=west] at (3.2,2.2) {{\small $=\Nout(3,1,3,4,2,2,4)$}};
	\node[anchor=west] at (9.2,24.2) {{\small $X(3,1,0,2,4,0)$}};
	\node[anchor=west] at (9.2,22.2) {{\small $=\Nout(2,3,1,0,2,4,0)$}};
	\node[anchor=west] at (9.2,19.2) {{\small $X(3,1,0,2,4,1)$}};
	\node[anchor=west] at (9.2,17.2) {{\small $=\Nout(2,3,1,0,2,4,1)$}};
	\node[anchor=west] at (9.2,9.2) {{\small $X(3,1,0,2,4,3)$}};
	\node[anchor=west] at (9.2,7.2) {{\small $=\Nout(2,3,1,0,2,4,3)$}};
	\node[anchor=west] at (9.2,4.2) {{\small $X(3,1,0,2,4,4)$}};
	\node[anchor=west] at (9.2,2.2) {{\small $=\Nout(2,3,1,0,2,4,4)$}};
	\node[draw=black,fill=black,inner sep=2] at (3,10) {};
	\node[draw=black,fill=black,inner sep=2] at (3,8) {};
	\node[draw=black,fill=black,inner sep=2] at (3,7) {};
	\node[draw=black,fill=black,inner sep=2] at (3,6) {};
	\node[draw=black,fill=black,inner sep=2] at (9,15) {};
	\node[draw=black,fill=black,inner sep=2] at (9,14) {};
	\node[draw=black,fill=black,inner sep=2] at (9,13) {};
	\node[draw=black,fill=black,inner sep=2] at (9,11) {};
	\end{tikzpicture}
	\caption{Illustration of the construction of the power dominating set $S$ for $d=5$ and $n=7$. For
		the two columns $(a_1,\ldots,a_5)=(1,3,4,4,2)$ and $(a_1,\ldots,a_5)=(3,1,0,2,4)$ we show the
		elements of $S$ (black squares), and we indicate for the sets $X(a_1,\ldots,a_6)$ (enclosed by
		rectangles) the elements of $S$ having them as their out-neighbourhood.}
	\label{fig:debruijn}
\end{figure}
\begin{lemma}\label{lem:power_domination}
  The set $S$ defined in~(\ref{eq:pd_set_debruijn}) is a power dominating set for $G$.
\end{lemma}
\begin{proof}
  For $Z=S\cup \Nout(S)$ it is sufficient to show that
  $\lvert Z\cap X(a_1,\ldots,a_{n-1})\rvert\geqslant d-1$ for every $(a_1,\ldots,a_{n-1})$. We
  provide the full argument for $n\geqslant 4$ (the cases $n=2$ and $n=3$ are easy to check).
  \begin{description}
  \item[Case 1.] If $a_{n-1}=a_1+a_{n-2}$, then by~(\ref{eq:pd_set_debruijn}),
    \[S\cap X(a_1,\ldots,a_{n-1})=\{\left(a_1,\ldots,a_n\right)\ :\ a_n\in \ints_d\setminus\{
      a_1+a_2+a_{n-2}\}\},\]
    hence $\lvert Z\cap X(a_1,\ldots,a_{n-1})\rvert\geq \lvert S\cap X(a_1,\ldots,a_{n-1})\rvert=d-1$.
  \item[Case 2.]  If $a_{n-1}\neq a_1+a_{n-2}$, then $X(a_1,\ldots,a_{n-1})\subseteq Z$ because
    \[X(a_1,\ldots,a_{n-1})=\Nout((a_{n-2}-a_{n-3},a_1,a_2,\ldots,a_{n-1}))\]
    and $(a_{n-2}-a_{n-3},a_1,a_2,\ldots,a_{n-1})\in S$. \qedhere
  \end{description}
\end{proof}
The second part of Theorem~\ref{main:thm:debruijn} follows from Lemmas~\ref{lb:pddebruijn} and~\ref{lem:power_domination}.

\section{The power domination number of Kautz digraphs}\label{sec:kautz}
In this section we prove Theorem~\ref{main:thm:kautz}. Let us define the sets
\[X(a_1,\ldots,a_{n-1}) = \left\{(a_1,\ldots,a_{n-1},a_n)\ :\ a_n\in\ints_{d+1}\setminus\{a_{n-1}\}\right\}\]
for $(a_1,\ldots,a_{n-1})\in\ints_{d+1}^{n-1}$ with $a_i\neq a_{i+1}$ for all $i$. These sets
partition the vertex set $V$ into $(d+1)d^{n-2}$ sets of size $d$. Furthermore,
$\Nout(v)= X(a_1,\ldots,a_{n-1})$ for every vertex $v$ of the form $(a_0,a_1,a_2,\ldots,a_{n-1})$.
\begin{lemma}\label{lb:zfkautz}
  Let $G$ be a Kautz digraph with parameters $d,n\geqslant 2$. Then $Z(G) \geq (d-1)(d+1)d^{n-2}$.
\end{lemma}
\begin{proof}
  Every $2$-element subset of each of the sets $X(a_1,\ldots,a_{n-1})$ is strongly critical, and
  therefore, any zero forcing set $S$ needs to intersect $X(a_1,\ldots,a_{n-1})$ in at least $d-1$
  elements, and the result follows.
\end{proof}

\begin{lemma}\label{ub:zfkautz}
 Let $G$ be a Kautz digraph with parameters $d,n\geqslant 2$. Then $Z(G) \leq (d-1)(d+1)d^{n-2}$.
\end{lemma}
\begin{proof} 
  Consider the vertex set
\begin{equation*}
S=
\begin{cases}
\{(a_1,a_2)\in V\ :\ a_2\neq a_1+1\} & \text{if }n=2,\\
\{\left(a_1,\ldots,a_n\right)\in V\ :\ a_{n}\neq a_{n-2}\} & \text{if }n\geqslant 3.    
\end{cases}
\end{equation*}
 We have $\lvert S\rvert=(d-1)(d+1)d^{n-2}$, and to show that $S$ is
  a zero forcing set, it is sufficient to verify that each vertex $v=(a_1,\ldots,a_{n-1},a_n)$ is
  either in $S$ or is the unique out-neighbor in $V\setminus S$ for some vertex $w$.
  \begin{description}
  \item[Case $n=2$.] If $a_2\neq a_1+1$ then $v\in s$. If $a_2= a_1+1$ then for any vertex of the
  form $w=(\beta,a_1)$, $v$ is the only neighbor of $w$ in $V\setminus S$. 
  \item[Case $n\geq 3$.] If $a_n \neq a_{n-2}$, then $v \in S$. If $a_n = a_{n-2}$, then for any vertex of the
  form $w=(\beta,a_1,\ldots,a_{n-1})$, $v$ is the only neighbor of $w$ in $V\setminus S$. \qedhere
  \end{description}
\end{proof}
Lemmas~\ref{lb:zfkautz} and~\ref{ub:zfkautz} imply the first statement of Theorem~\ref{main:thm:kautz}.
\begin{lemma}\label{lb:pdkautz}
  Let $G$ be a Kautz digraph with parameters $d\geq 2$ and $n \geq 3$. Then, every power dominating
  set has size at least $(d-1)(d+1)d^{n-3}$.
\end{lemma}
\begin{proof}
  Let $S$ be a power-dominating set, suppose $\lvert S\rvert<(d-1)(d+1)d^{n-3}$ and set
  $Z=S\cup \Nout(S)$. We have
  \[(Z\setminus S)\cap X(a_1,\ldots,a_{n-1})\neq\emptyset\implies X(a_1,\ldots,a_{n-1})\subseteq
    Z.\]
  For $k=0,1,\ldots,d$, we set
  $\alpha_k=\#\{(a_1,\ldots,a_{n-1})\ :\ \lvert S\cap X(a_1,\ldots,a_{n-1})\rvert=k\}$, and get
  \[\lvert S\rvert=\alpha_1+2\alpha_2+\cdots+(d-1)\alpha_{d-1}+d\alpha_d.\]
  Now let $I_0=\{(a_1,\ldots,a_{n-1})\ :\ X(a_1,\ldots,a_{n-1})\subseteq Z\}$. Clearly,
  \[\lvert I_0\rvert\leqslant \lvert S\rvert+\alpha_d=\alpha_1+2\alpha_2+\cdots+(d-1)\alpha_{d-1}+(d+1)\alpha_d.\]
  For $(a_1,\ldots,a_{n-1})\notin I_0$ we must have $\lvert Z\cap X(a_1,\ldots,a_{n-1})\rvert=d-1$
  because $Z$ intersects every weakly critical set. This implies that
  $\lvert S\cap X(a_1,\ldots,a_{n-1})\rvert=d-1$, and we conclude
  $\lvert I_0\rvert+\alpha_{d-1}\geqslant (d+1)d^{n-2}$. Therefore
  \[\alpha_1+2\alpha_2+\cdots+(d-2)\alpha_{d-2}+d\alpha_{d-1}+(d+1)\alpha_d\geqslant
    (d+1)d^{n-2},\]
  and together with $\lvert S\rvert<(d-1)(d+1)d^{n-3}$ this yields
  \[\alpha_{d-1}+\alpha_d>(d+1)d^{n-2}-(d-1)(d+1)d^{n-3}=(d+1)d^{n-3}.\]
  But then $\lvert S\rvert\geqslant (d-1)(\alpha_{d-1}+\alpha_d)>(d-1)(d+1)d^{n-3}$, which is
  the required contradiction.
\end{proof}
We define a set $S\subseteq V$ by
\begin{equation}\label{eq:pd_set_kautz}
S=
\begin{cases}
\{(0,1),(0,2),\ldots,(0,d)\} & \text{if }n=2,\\
\{(a_1,a_2,a_3)\in V\ :\ a_2=a_1+1,\,a_3\neq a_1+2\} & \text{if }n=3,\\
\{(a_1,a_2,a_3,a_4)\in V\ :\ a_3=a_1,\,a_4\neq a_2\} & \text{if }n=4,\\
\left\{\left(a_1,\ldots,a_n\right)\in V\,:\,\left((a_{n-2},a_{n-1})=(a_1,a_2)\land a_n\neq
    a_3\right)\lor \left(a_{n-1}=a_1\land a_n\neq
    a_2\right)\right\} & \text{if }n\geqslant 5.    
\end{cases}
\end{equation}
\begin{lemma}\label{lem:size_of_S}
$\displaystyle\lvert S\rvert=
\begin{cases}
d & \text{if }n=2,\\
(d-1)(d+1)d^{n-3}& \text{if }n\geqslant 3.    
\end{cases}
$
\end{lemma}
\begin{proof}
  For $n\leq 4$ this is easy to check. For $n\geq 5$ we proceed by the following argument. We
  consider the partition $S=S_1\cup S_2$ where
  \begin{align*}
    S_1 &= \{(a_1,\ldots,a_n)\in S\ :\ a_{n-3}=a_1\}, & S_2 &= \{(a_1,\ldots,a_n)\in S\ :\
                                                              a_{n-3}\neq a_1\}.
  \end{align*}
  Let $s_k$ be the number of words $a_1\ldots a_k$ over the alphabet $\ints_{d+1}$ which satisfy
  $a_k=a_1$ and $a_i\neq a_{i+1}$ for all $i\in\{1,\ldots,k-1\}$. Then $s_2=0$ and
  $s_k=(d+1)d^{k-2}-s_{k-1}$ for $k\geq 3$. It follows by induction on $k$ that $s_k=d^{k-1}-(-1)^kd$. Every vector
  $(a_1,\ldots,a_{n-3})\in\ints_{d+1}^{n-3}$ with $a_i\neq a_{i+1}$ and $a_{n-3}=a_1$ can be
  extended to an element of $S_1$ by choosing $a_{n-2}\in\ints_{d+1}\setminus\{a_1\}$, $a_{n-1}=a_1$
  and $a_{n}\in\ints_{d+1}\setminus\{a_1,a_2\}$, hence
  \[\lvert S_1\rvert = s_{n-3}d(d-1)=\left(d^{n-4}-(-1)^{n-3}d\right)d(d-1).\]
  If $a_{n-3}\neq a_1$ then we can choose $(a_{n-2},a_{n-1})=(a_1,a_2)$ and
  $a_n\in\ints_{d+1}\setminus\{a_2,a_3\}$, or $a_{n-2}\in\ints_{d+1}\setminus\{a_1,a_{n-3}\}$,
  $a_{n-1}=a_1$ and $a_n=\ints_{d+1}\setminus\{a_1,a_2\}$, hence
  \begin{align*}
    \lvert S_2\rvert &= \left[(d+1)d^{n-4}-s_{n-3}\right]\left[(d-1)+(d-1)^2\right] \\
                     &= \left[(d+1)d^{n-4}-d^{n-4}+(-1)^{n-3}d\right]d(d-1)\\
                     &= \left[d^{n-3}+(-1)^{n-3}d\right]d(d-1). 
  \end{align*}
  Finally,
  \[\lvert S\rvert=\lvert S_1\rvert+\lvert S_2\rvert=d(d-1)\left[d^{n-4}-(-1)^{n-3}d+d^{n-3}+(-1)^{n-3}d\right]=(d+1)(d-1)d^{n-3}.\qedhere\]
\end{proof}

\begin{lemma}\label{lem:power_domination_kautz}
  The set $S$ defined in~(\ref{eq:pd_set_kautz}) is a power dominating set for $G=\K(d,n)$.
\end{lemma}
\begin{proof}
  For $Z=S\cup \Nout(S)$ it is sufficient to show that
  $\lvert Z\cap X(a_1,\ldots,a_{n-1})\rvert\geqslant d-1$ for every $(a_1,\ldots,a_{n-1})$. We
  provide the full argument for $n\geqslant 5$ (the cases $n=2$, $n=3$ and $n=4$ are easy to check).
  \begin{description}
  \item[Case 1.] If $a_{n-2}=a_1$ and $a_{n-1}=a_2$ then
    \[\lvert S\cap X(a_1,\ldots,a_{n-1})\rvert=\lvert\{\left(a_1,\ldots,a_n\right)\ :\ a_n\in
      \ints_{d+1}\setminus\{a_2, a_3\}\}\rvert=d-1,\]
    and the claim follows from $Z\supseteq S$.
  \item[Case 2.] If $a_{n-2}=a_1$ and $a_{n-1}\neq a_2$, then $X(a_1,\ldots,a_{n-1})\subseteq Z$
    because
    \[X(a_1,\ldots,a_{n-1})=\Nout((a_{n-3},a_1,a_2,\ldots,a_{n-1}))\]
    and $(a_{n-3},a_1,a_2,\ldots,a_{n-1})\in S$.
  \item[Case 3.] If $a_{n-2}\neq a_1$ and $a_{n-1}= a_2$, then $X(a_1,\ldots,a_{n-1})\subseteq Z$
    because
    \[X(a_1,\ldots,a_{n-1})=\Nout((a_{n-2},a_1,a_2,\ldots,a_{n-1}))\] and
    $(a_{n-2},a_1,a_2,\ldots,a_{n-1})\in S$.
  \item[Case 4.] If $a_{n-2}\neq a_1$ and $a_{n-1} = a_1$ then
    \[\lvert S\cap X(a_1,\ldots,a_{n-1})\rvert=\lvert\{\left(a_1,\ldots,a_n\right)\ :\ a_n\in
      \ints_{d+1}\setminus\{a_1,a_2\}\}\rvert=d-1,\]
    and the claim follows from $Z\supseteq S$.
  \item[Case 5.] If $a_{n-2}\neq a_1$ and $a_{n-1} \not\in \{a_1,a_2\}$, then
    $X(a_1,\ldots,a_{n-1})\subseteq Z$ because
    \[X(a_1,\ldots,a_{n-1})=\Nout((a_{n-2},a_1,a_2,\ldots,a_{n-1}))\] and
    $(a_{n-2},a_1,a_2,\ldots,a_{n-1})\in S$. \qedhere
	\end{description}
\end{proof}
The second part of Theorem~\ref{main:thm:kautz} follows from Lemmas~\ref{lb:pdkautz},~\ref{lem:size_of_S} and~\ref{lem:power_domination_kautz}.

\section{Conclusion}

In this paper, we have determined the zero forcing number and power domination number of de Bruijn and
Kautz digraphs. There are many variants of de Bruijn and Kautz digraphs introduced and studied over
the years, one of them being generalized de Bruijn digraphs $GB(d,n)$ and generalised Kautz digraphs
$GK(d,n)$ which can be defined as follows:
\begin{align*}
V(GB(d,n)) &= \{0,1, \ldots, n-1\}, \\
A(GB(d,n)) &= \left\{(x,y)\ :\ y \equiv dx +i \pmod{n},\ 0\leq i \leq d-1\right\},\\
V(GK(d,n))&=\{0,1, \ldots, n-1\},\\
A(GK(d,n))&=\left\{(x,y)\ ;\ y \equiv -dx -i \pmod{n},\ 1\leq i \leq d\right\}.
\end{align*}
We leave it as an open problem to determine the zero forcing number and power domination number of generalised de Bruijn and Kautz digraphs.


\begin{thebibliography}{10}
	
\bibitem{AazamiStilp2009} A.~Aazami, and K.~Stilp.  \newblock Approximation Algorithms and Hardness
  for Domination with Propagation.  \newblock {\em SIAM Journal on Discrete Mathematics},
  23:1382-1399, 2009.
	
\bibitem{AIM-2008-Zeroforcingsets} {AIM Minimum Rank -- Special Graphs Work Group}.  \newblock Zero
  forcing sets and the minimum rank of graphs.  \newblock {\em Linear Algebra and its Applications},
  428(7):1628--1648, 2008.
	
\bibitem{BarioliBarrettFallatEtAl2010} F.~Barioli, W.~Barrett, S.~M. Fallat, H.~T. Hall, L.~Hogben,
  B.~Shader, P.~van~den Driessche, and H.~van~der Holst.  \newblock Zero forcing parameters and
  minimum rank problems.  \newblock {\em Linear Algebra and its Applications}, 433(2):401--411,
  2010.
	
\bibitem{Barioli2009} F.~Barioli, S.~M. Fallat, H.~T. Hall, D.~Hershkowitz, L.~Hogben, H.~van~der
  Holst, and B.~Shader, \newblock On the minimum rank of not necessarily symmetric matrices: a
  preliminary study.  \newblock {\em Electronic Journal of Linear Algebra,}, 18:126--145, 2009.
	
\bibitem{BarioliBarrettFallatHallHogbenShaderDriesscheHolst-2012-ParametersRelatedto} F.~Barioli,
  W.~Barrett, S.~M. Fallat, H.~T. Hall, L.~Hogben, B.~Shader, P.~van~den Driessche, and H.~van~der
  Holst.  \newblock Parameters related to tree-width, zero forcing, and maximum nullity of a graph.
  \newblock {\em Journal of Graph Theory}, 72(2):146--177, 2012.
	
\bibitem{BermanFriedlandHogbenEtAl2008} A.~Berman, S.~Friedland, L.~Hogben, U.~G. Rothblum, and
  B.~Shader.  \newblock An upper bound for the minimum rank of a graph.  \newblock {\em Linear
    Algebra and its Applications}, 429(7):1629--1638, 2008.
	
\bibitem{Chang2012} G.~J.~Chang, P.~Dorbec, P.~Montassier, and A.~Raspaud.  \newblock Generalized
  power domination of graphs.  \newblock {\em Discrete Applied Mathematics}, 160:1691--1698, 2012.
	
\bibitem{Dorbec2008} P.~Dorbec, M.~Mollard, S.~Klavzar, and S.~Spacapan.  \newblock Power Domination
  in Product Graphs.  \newblock {\em SIAM Journal on Discrete Mathematics}, 22:554--567, 2008.
		
\bibitem{Dong2015a} Y.~Dong, E.~Shan, and L.~Kang \newblock Constructing the minimum dominating sets
  of generalized de Bruijn digraphs.  \newblock {\em Discrete Mathematics}, 338:1501--1508, 2015.
	
\bibitem{EdholmHogbenHuynhEtAl2012} C.~J. Edholm, L.~Hogben, M.~Huynh, J.~LaGrange, and D.~D. Row.
  \newblock Vertex and edge spread of zero forcing number, maximum nullity, and minimum rank of a
  graph.  \newblock {\em Linear Algebra and its Applications}, 436(12):4352--4372, 2012.  \newblock
  Special Issue on Matrices Described by Patterns.
	
\bibitem{Haynes2002a} T.~W. Haynes, S.~M. Hedetniemi, S.~T. Hedetniemi, and M.~A. Henning.
  \newblock Domination in graphs applied to electric power networks.  \newblock {\em {SIAM}
    J. Discrete Math.}, 15(4):519--529, 2002.
	
\bibitem{HogbenHuynhKingsleyMeyerWalkerYoung-2012-Propagationtimezero} L.~Hogben, M.~Huynh,
  N.~Kingsley, S.~Meyer, S.~Walker, and M.~Young.  \newblock Propagation time for zero forcing on a
  graph.  \newblock {\em Discrete Applied Mathematics}, 160(13-14):1994--2005, 2012.
	
	
\bibitem{Huang2006} J. Huang, and J.-M. Xu.  \newblock The bondage numbers of extended de Bruijn and
  Kautz digraphs.  \newblock {\em Computers and Mathematics with Applications}, 51:1137--1147, 2006.
	
\bibitem{Imase1983} M. Imase, and M. Itoh.  \newblock A Design for Directed Graphs with Minimum Diameter.
  \newblock {\em IEEE Transactions on Computers, Institute of Electrical and Electronics Engineers
    (IEEE),}, 32:782--784, 1983.
		
\bibitem{Kuo2015} J.~Kuo and W.-L. Wu.  \newblock Power domination in generalized undirected de
  Bruijn graphs and Kautz graphs.  \newblock {\em Discrete Math. Algorithm. Appl.}, 07:2961--2973,
  2015.
		
\bibitem{Lu.etal_2015_Proofconjecturezero} L.~{Lu}, B.~{Wu}, and Z.~{Tang}.  \newblock Proof of a
  conjecture on the zero forcing number of a graph.  \newblock {\em
    arXiv:\href{http://arxiv.org/abs/1507.01364}{1507.01364}}, 2015.
	
\bibitem{Severini2008} S.~Severini.  \newblock Nondiscriminatory propagation on trees.  \newblock
  {\em Journal of Physics A: Mathematical and Theoretical}, 41(48):482002, 2008.
	
\bibitem{Stephen2015} S.~Stephen, B.~Rajan, J.~Ryan, C.~Grigorious, and A.~William.  \newblock Power
  domination in certain chemical structures.  \newblock {\em Journal of Discrete Algorithms},
  33:10--18, 2015.
	
\end{thebibliography}

\end{document}